\documentclass[11pt,letterpaper]{article}

\usepackage{amsmath,amsfonts,amsthm,amssymb,stmaryrd,graphicx,relsize,bm}  
\theoremstyle{plain}

\numberwithin{equation}{section}

\newtheorem{thm}{Theorem}[section]
\newtheorem{lem}[thm]{Lemma}

\newenvironment{exam}[1]
{\begin{flushleft}\textbf{Example #1}.\enspace}%
{\end{flushleft}}

\allowdisplaybreaks  

\newcommand{\trace}{tr}
\newcommand{\instr}{In}
\newcommand{\ob}{Ob}
\newcommand{\ityes}{\textit{yes}}
\newcommand{\itno}{\textit{no}}

\newcommand{\escript}{\mathcal{E}}
\newcommand{\hscript}{\mathcal{H}}
\newcommand{\iscript}{\mathcal{I}}
\newcommand{\jscript}{\mathcal{J}}
\newcommand{\kscript}{\mathcal{K}}
\newcommand{\lscript}{\mathcal{L}}
\newcommand{\oscript}{\mathcal{O}}
\newcommand{\sscript}{\mathcal{S}}

\newcommand{\hscripthat}{\widehat{\hscript}}
\newcommand{\iscripthat}{\widehat{\iscript}}
\newcommand{\jscripthat}{\widehat{\jscript}}
\newcommand{\kscripthat}{\widehat{\kscript}}

\newcommand{\iscriptbar}{\overline{\iscript}}

\newcommand{\kscriptbar}{\overline{\kscript}}

\newcommand{\ab}[1]{\left|#1\right|}
\newcommand{\brac}[1]{\left\{#1\right\}}
\newcommand{\paren}[1]{\left(#1\right)}
\newcommand{\sqbrac}[1]{\left[#1\right]}

\begin{document}

\title{MULTI-OBSERVABLES AND\\MULTI-INSTRUMENTS}
\author{Stan Gudder\\ Department of Mathematics\\
University of Denver\\ Denver, Colorado 80208\\
sgudder@du.edu}
\date{}
\maketitle

\begin{abstract}
This article introduces the concepts of multi-observables and multi-instruments in quantum mechanics. A multi-observable $A$ (multi-instrument $\iscript$) has an outcome space of the form $\Omega =\Omega _1\times\cdots\times\Omega _n$ and is denoted by $A_{x_1\cdots x_n}$ ($\iscript _{x_1\cdots x_n}$) where
$(x_1,\ldots ,x_n)\in\Omega$. We also call $A$ ($\iscript$) an $n$-observable ($n$-instrument) and when $n=2$ we call $A$ ($\iscript$) a bi-observable (bi-instrument).
We point out that bi-observables $A$ ($\iscript$) and bi-instruments have been considered in past literature, but the more general case appears to be new. In particular, two observables (instruments) have been defined to coexist or be compatible if they possess a joint bi-observable (bi-instrument). We extend this definition to $n$ observables and $n$ instruments by considering joint marginals of $n$-observables and joint reduced marginals of $n$-instruments. We show that a $n$-instrument measures a unique $n$-observable and if a finite umber of instruments coexist, then their measured observables coexist. We prove that there is a close relationship between a nontrivial $n$-observable and its parts. Moreover, a similar result holds for instruments. We next show that a natural definition for the tensor product of a finite number of instruments exist and possess reasonable properties. We then discuss sequential products of a finite number of observables and instruments. We present various examples such as Kraus, Holevo and L\"uders instruments.
\end{abstract}

\section{Introduction}  
In this introduction we speak in general terms and the precise definitions are given in Section~2. Quantum mechanics can be described as a theory of measurements
\cite{bgl95,dl70,hz12,lud51,nc00}. Mathematically, these measurements are given by observables and instruments. An observable $A$ corresponds to an experiment that has various outcomes and when an outome is $x$, we say that the effect $A_x$ occurs. An instrument is more general and gives more information about a quantum system than an observable. We think of an instrument as an apparatus that can be employed to measure a unique observable and also updates the state of the system when an outcome is registered \cite{bgl95,dl70,hz12}.

A multi-observable is an observable $A$ whose outcome space has the form $\Omega _A=\Omega _1\times\cdots\times\Omega _n$. We also say that $A$ is an
$n$-observable. If $(x_1,\ldots ,x_n)\in\Omega _A$ is an outcome for $A$, we denote the corresponding physical effect by $A_{x_1\cdots x_n}$. The $i$-marginal of $A$ is defined as the observable with outcome space $\Omega _i$ and given by
\begin{equation*}
A_y^i=\sum\brac{A_{x_1\cdots x_{(i-1)}yx_{(i+1)}\cdots x_n}\colon x_1,\ldots ,x_{(i-1)},x_{(i+1)},\cdots x_n}
\end{equation*}
Thus $A_y^i$ is defined to be the sum of $A_{x_1\cdots x_n}$ whose $i$th component is $y$. In a similar way, we define a multi-instrument (or $n$-instrument) to be an instrument $\iscript$ whose outcome space has the form $\Omega _\iscript=\Omega _1\times\cdots\Omega _n$ and the $i$-marginal of $\iscript$ is similar. If $n=2$, we call $A(\iscript )$ a bi-observable (bi-instrument). Although the idea of a general $n$-observable or $n$-instrument appears to be new, bi-observables and bi-instruments have been considered in the past \cite{dapt22,bkmpt22,ls22,mt122,mt222}. In particular, two observables $A,B$ on the same Hilbert space are said to coexist (be compatible) if there exists a bi-observable $C$ such that its marginals satisfying $C^1=A$, $C^2=B$. In this way, $A$ and $B$ are simultaneously measurable using $C$. We now extend this idea by defining observables $B_i$, $i=1,2,\ldots ,n$, to coexist (be compatible) if there exists a joint $n$-observable $C$ with marginals $C^i=B_i$, $i=1,2,\ldots, n$. In particular, the marginals of an $n$-observable coexist.

For observables $A_i$, $i=1,2,\ldots ,n$, on Hilbert space $H_i$, we define their tensor product to be the $n$-observable $A$ given by
\begin{equation*}
A_{x_1\cdots x_n}=A_{1x_1}\otimes\cdots\otimes A_{nx_n}
\end{equation*}
The marginals become
\begin{equation*}
A_{x_i}^i=I_1\otimes\cdots\otimes I_{i-1}\otimes A_{x_i}\otimes I_{i+1}\otimes\cdots\otimes I_n
\end{equation*}
where $I_j$ is the identity operator on the $j$th Hilbert space $H_j$. We obtain the reduced observables $^{i\!\!}A$ corresponding to $A$ by taking the partial traces of $A$ relative to $H_j$, $j\ne i$, and dividinhg by the product of $\dim (H_j)$, $j\ne i$. We then obtain $^{i\!\!}A^i=A_i$, $i=1,2,\ldots ,n$. For a general instrument on
$H_1\otimes\cdots\otimes H_n$, the reduced instrument $^{i\!\!}\iscript$ corresponding to $\iscript$ is obtained by taking the partial traces of $\iscript$ relative to
$H_j$, $j\ne i$. We say that instruments $\iscript _i$, $i=1,2,\ldots ,n$, coexist (are compatible) if there exists an $n$-instrument $\iscript$ on
$H_1\otimes\cdots\otimes H_n$ such $^{\!\!i}\iscript ^i=\iscript _i$, $i=1,2,\ldots ,n$, where the marginal $\iscript ^i$ is defined like the marginal of an observable.

A positive operator $\rho$ on a Hilbert space $H$ with trace $\trace (\rho )=1$ is called a state \cite{bgl95,dl70,hz12,nc00}. States are employed to describe the condition of a quantum system and the set of a states on $H$ is denoted by $\sscript (H)$. If $\rho\in\sscript (H)$ and $A$ is an observable on $H$ then $0\le A_x\le I$ is an operator on $H$ and $\trace (\rho A_x)$ is the probability that a measurement of $A$ results in outcome $x$ when the system is in state $\rho$.
We call $\Phi _\rho ^A(x)=\trace (\rho A_x)$ the $\rho$-distribution of $A$. If $\iscript$ is an instrument on $H$, then $\iscript _x(\rho )$ gives the (unnormalized) updated stat when a measurement of $\iscript$ results in outcome $x$. The $\rho$-distribution of $\iscript$ is given by
$\Phi _\rho ^\iscript (x)=\trace\sqbrac{\iscript _x(\rho )}$. If $\Omega _\iscript =\Omega _A$ and $\Phi _\rho ^\iscript (x)=\Phi _\rho ^A(x)$ for all
$x\in\Omega _\iscript$,
$\rho\in\sscript (H)$, we say that $\iscript$ measures the observable $A$. It easily follows that an instrument $\iscript$ measures a unique observable $\iscripthat$. However, an observable is measured by many instruments. In Section~2 we show that if instruments $\iscript _i$, $i=1,2,\ldots ,n$, coexist, then $\iscripthat _i$ also coexist. The converse of this statement does not hold
\cite{hz12, nc00}. We also consider parts of observables and instruments in Section~2 \cite{gud220,hz12}.

In section~3, we first show that there is a close relationship between a nontrivial $n$-observable and its parts. A similar result holds for instruments. We next show that a natural definition of the tensor product of a finite number of instruments exists and then prove that this definition possesses reasonable properties. We then discuss sequential products of a finite number of observables and instruments. We present various examples such as Kraus, Holevo and L\"uders instruments
\cite{hol82,kra83,lud51}.
\section{Basic Definitions and Concepts}  
All the Hilbert spaces in this article are assumed to be finite-dimensional. For a Hilbert space $H$, we denote the set of linear operators on $H$ by $\lscript (H)$ and the set of self-adjoint operators on $H$ by $\lscript _S(H)$. The zero and identity operators are designated by $0,I$, respectively. When the Hilbert space needs to be specified we write $I_H$ instead of $I$. An operator $a\in\lscript _S(H)$ that satisfies $0\le a\le I$ is called an \textit{effect} and the set of effects on $H$ is denoted by
$\escript (H)$ \cite{bgl95,dl70,hz12}. We consider an effect as a two-valued \ityes -\itno\ experiment and when the value \ityes\ is obtained, then $a$ \textit{occurs}
\cite{bgl95,dl70,hz12}. If $\Omega _A$ is a finite set, then a set of effects $A=\brac{A_x\colon x\in\Omega _A}\subseteq\escript (X)$ that satisfies
$\sum\limits _{x\in\Omega _A}A_x=I$ is called an \textit{observable}. The set of observables on $H$ is denoted by $\ob (H)$. We call $\Omega _A$ the
\textit{outcome space} for $A$ and when $\Delta\subseteq\Omega _A$, the map $A(\Delta )=\sum\limits _{x\in\Delta}A_x$ is an \textit{effect-valued measure}
(or \textit{positive operator-valued measure} (POVM)) \cite{bgl95, dl70, hz12}. A \textit{state} on $H$ is a positive operator $\rho$ with trace $\trace (\rho )=1$
\cite{bgl95,dl70,hz12,nc00} and the set of states on $H$ is denoted by $\sscript (H)$. If $\rho\in\sscript (H)$, $a\in\escript (H)$ then $0\le\trace (\rho a)\le 1$ and we call $P_\rho (a)=\trace (\rho a)$ the \textit{probability that} $a$ \textit{occurs}. Of course, $P_\rho (0)=0$ and $P_\rho (I)=1$ for all $\rho\in\sscript (H)$ so $0$ never occurs and $I$ always occurs. For $A\in\ob (H)$, $\rho\in\sscript (H)$, the probability measure on $\Omega _A$ given by
\begin{equation*}
\Delta\mapsto P_\rho ^A(\Delta )=\trace\sqbrac{\rho A(\Delta )}
\end{equation*}
is called the $\rho$-\textit{distribution} of $A$ in the state $\rho$ \cite{bgl95,dl70,hz12,nc00}. A set of effects $\brac{a_i\colon i=1,2,\ldots ,n}\subseteq\escript (H)$
\textit{coexist} (are \textit{compatible}) if there exists an observable $A\in\ob (H)$ such that $a_i=A(\Delta _i)$, $\Delta _i\subseteq\Omega _A$, $i=1,2,\ldots ,n$. It can be shown that if a set of effects mutually commute, then they coexist but the converse does not hold \cite{hz12,nc00}. We call the observable $A$ a
\textit{joint observable} for $\brac{a_i\colon i=1,2,\ldots ,n}$. Clearly, any subset of the effects of an observable coexist. An observable $B$ is \textit{part} of an observable $A$ if there exists a surjection $f\colon\Omega _A\to\Omega _B$ such that
\begin{equation*}
B_y=A\sqbrac{f^{-1}(y)}=\sum\brac{A_x\colon f(x)=y}
\end{equation*}
and we write $B=f(A)$ \cite{gud220,hz12}. In this case, we have $\Omega _B=f(\Omega _A)$. We denote the cardinality of a finite set $S$ by $\ab{S}$. If
$\ab{f(\Omega _A)}=1$, in which case $B=f(A)$ is a trivial observable $B_y=A(\Omega _A)=I$ and we call $B$ a \textit{trivial part} of $A$ \cite{gud220}.

An observable $A$ is a \textit{multi-observable} if $\Omega _A$ can be arranged to be a product set $\Omega _1\times\cdots\times\Omega  _n$ and we then call
$A$ an $n$-\textit{observable}. To be  precise, $A$ is an $n$-\textit{observable} if there exists a bijection
$h\colon\Omega _A\to\Omega _1\times\cdots\times\Omega _n$, $\Omega _i\ne\emptyset$, $i=1,2,\ldots ,n$ and we write
$\Omega _A\approx\Omega _1\times\cdots\times\Omega _n$. Any observable is trivially an $n$-observable because we can let $\Omega _i=\brac{1}$,
$i=1,2,\ldots ,n-1$ and then 
\begin{equation*}
\Omega _A\approx\Omega _1\times\cdots\times\Omega _{n-1}\times\Omega _A
\end{equation*}
We say that $A$ is a \textit{nontrivial} $n$-\textit{observable} if $\Omega _A\approx\Omega _1\times\cdots\times\Omega _n$ where $\ab{\Omega _i}\ge 2$, $i=1,2,\ldots ,n$. Let $A\in\ob (H)$ be a nontrivial $n$-observable in which case we can assume that $\Omega _A=\Omega _1\times\cdots\times\Omega _n$ where
$\ab{\Omega _i}\ge 2$, $i=1,2,\ldots ,n$. Letting $y\in\Omega _i$, we use the notation
\begin{equation*}
\Omega _A^{(x_i=y)}=\Omega _1\times\cdots\times\Omega _{i-1}\times\brac{y}\times\Omega _{i+1}\times\cdots\times\Omega _n
\end{equation*}
We define the $i$-\textit{marginal} of $A$ to be the observable $A^i\in\ob (H)$ such that $\Omega _{A^i}=\Omega _i$ given by
\begin{equation*}
A_y^i=A\sqbrac{\Omega _A^{(x_i=y)}}
\end{equation*}
Thus, $A_y^i$ sums all the effects in $A$ whose $i$th component is $y$. We now show that $A^i$ is part of $A$, $i=1,2,\ldots ,n$. Let
$f_i\colon\Omega _A\to\Omega _i$ be the surjection defined by
\begin{equation*}
f_i\paren{(x_1,\ldots ,x_n)}=x_i
\end{equation*}
$i=1,2,\ldots ,n$. Then for $y\in\Omega _i$ we have
\begin{equation}                
\label{eq21}
f_i(A)_y=A\sqbrac{f_i^{-1}(y)}=A\sqbrac{\brac{x\in\Omega _A\colon f_i(x)=y}}=A\sqbrac{\Omega _A^{(x_i=y)}}=A_y^i
\end{equation}
Hence, $A^i=f_i(A)$ so $A^i$ is part of $A$, $i=1,2,\ldots ,n$. We say that a set of observables $B_i\in\ob (H)$, $i=1,2,\ldots ,n$, \textit{coexist}
(are \textit{compatible}) if there exists a \textit{joint} $n$-observable $A\in\ob (H)$ such that the marginals $A^i=B_i$, $i=1,2,\ldots ,n$. In particular, the marginals of an $n$-observable coexist. Notice that if $B_1,\ldots ,B_n$ coexist, then the effects for $B_1,\ldots ,B_n$ coexist. Indeed, in this case we have that $B_i=A^i$ for a joint $n$-observable $A$. We then obtain $B_i=f_i(A)$ and for every $y\in\Omega _{B_i}$ we conclude that
\begin{equation*}
B_{iy}=A_y^i=f_i(A)=A\sqbrac{f_i^{-1}(y)}
\end{equation*}
Since this holds for all $i=1,2,\ldots ,n$, the effects for $B_1,\ldots ,B_n$ coexist.

We now consider some examples of $n$-observables.

\begin{exam}{1}  
Let $A_i\in\ob (H)$, $i=1,2,\ldots ,n$. The \textit{L\"uders sequential product} \cite{gn01,gud22,lud51} of these observables is the $n$-observable
$A\in\ob (H)$ with $\Omega _A=\Omega _{A_1}\times\cdots\times\Omega _{A_n}$ given by
\begin{equation*}
A_{x_1\cdots x_n}=(A_1\circ\cdots\circ A_n)_{x_1\cdots x_n}=A_{1x_1}^{1/2}\cdots A_{(n-1)x_{n-1}}^{1/2}A_{nx_n}A_{(n-1)x_{n-1}}^{1/2}\cdots A_{1x_1}^{1/2}
\end{equation*}
The marginals are
\begin{align*}
A_{x_1}^1&=A_{1x_1}\\
A_{x_2}^2&=\sum _{x_1\in\Omega _{A_1}}A_{1x_1}^{1/2}A_{2x_2}A_{1x_1}^{1/2}=\sum _{x_1\in\Omega _{A_1}}(A_1\circ A_2)_{x_1x_2}\\
&\ \vdots\\
A_{x_n}^n&=\sum _{x_1,\ldots ,x_{n-1}}A_{x_1}^{1/2}\cdots A_{(n-1)x_{n-1}}^{1/2}A_{nx_n}A_{(n-1)x_{n-1}}^{1/2}\cdots A_{1x_1}^{1/2}\\
   &=\sum _{x_1,\ldots ,x_{n-1}}(A_1\circ\cdots\circ A_n)_{x_1\cdots x_n}
\end{align*}
If $a\in\escript (H)$, we define the \textit{L\"uders map} $\lscript ^{(a)}(B)=a^{1/2}Ba^{1/2}$ for all $B\in\lscript (H)$ \cite{gn01,gud22,lud51}. The $\rho$-distribution of $A\in\ob (H)$ becomes
\begin{equation*}
\Phi _\rho ^A\paren{(x_1,\ldots ,x_n)}=\trace (\rho A_{x_1\cdots x_n})=\trace\sqbrac{\lscript ^{(A_{(n-1)x_{n-1}})}\cdots\lscript ^{(A_{1x_1})}(\rho )A_{nx_n}}
\hskip 1pc\qedsymbol
\end{equation*}
\end{exam}

\begin{exam}{2}  
For $A_i\in\ob (H_i)$, $i=1,2,\ldots ,n$, we define their \textit{tensor product} \cite{gud220} to be the $n$-observable $A\in\ob (H_1\otimes\cdots\otimes H_n)$ with
$\Omega _A=\Omega _{A_1}\times\cdots\times\Omega _{A_n}$ given by
\begin{equation*}
A_{x_1\cdots x_n}=A_{1x_1}\otimes\cdots\otimes A_{nx_n}
\end{equation*}
The marginals of $A$ become
\begin{align*}
A_{x_1}^1&=A_{1x_1}\otimes I_{H_2}\otimes\cdots I_{H_n}\\
A_{x_2}^2&=I_{H_1}\otimes A_{2x_2}\otimes I_{H_3}\otimes\cdots\otimes I_{H_n}\\
&\ \vdots\\
A_{x_n}^n&=I_{H_1}\otimes\cdots\otimes I_{H_{n-1}}\otimes A_{nx_n}
\end{align*}
Although the effects of $A$ need not commute and the effects of $A^i$ need not commute, we see that the effects of $A^i$ commute with those of $A^j$,
$i\ne j$. If $\rho\in\sscript (H_1\otimes\cdots\otimes H_n)$, the $\rho$-distribution of $A$ becomes
\begin{equation*}
\Phi _\rho ^A\paren{(x_1,\ldots ,x_n)}=\trace (\rho A_{1x_1}\otimes\cdots\times A_{nx_n})
\end{equation*}
When $\rho$ is a product state $\rho =\rho _1\otimes\cdots\otimes\rho _n$ we have
\begin{align*}
\Phi _\rho ^A\paren{(x_1,\ldots ,x_n)}&=\trace (\rho A_{1x_1})\trace (\rho A_{2x_2})\cdots\trace (\rho A_{nx_n})\\
   &=\Phi _{\rho _1}^{A_1}(x_1)\Phi _{\rho _2}^{A_2}(x_2)\cdots\Phi _{\rho _n}^{A_n}(x_n)
\end{align*}
We define the \textit{reduced} $n$-\textit{observable} $^{i\!\!}A\in\ob (H_i)$ with $\Omega _{(^{i\!\!}A)}=\Omega _A$, to be given by the partial trace
\begin{equation}                
\label{eq22}
^{i\!\!}A_{x_1\cdots x_n}=\tfrac{1}{m_1\cdots m_{i-1}m_{i+1}\cdots m_n}\,\trace _{H_1}\cdots\trace _{H_{i-1}}\trace _{H_{i_1}}\
   \cdots\trace _{H_n}(A_{x_1\cdots x_n})
\end{equation}
where $m_i=\dim (H_i)$, $i=1,2,\ldots ,n$. In this particular case, we have
\begin{align*}
&^{i\!\!}A_{x_1\cdots x_n}\\
&\ =\tfrac{1}{m_1\cdots m_{i-1}m_{i+1}\cdots m_n}\,\trace (A_{1x_1})\cdots\trace (A_{(i-1)x_{i-1}})\trace (A_{(i+1)x_{i+1}})
   \cdots\trace (A_{nx_n})A_{ix_i}
\end{align*}
We then obtain the reduced marginals $^{i\!\!}A_{x_i}^i=A_{ix_i}$ so $^{i\!\!}A^i=A_i$, $i=1,2,\ldots ,n$. We also have the \textit{mixed reduced marginals}
$^{i\!\!}A^j\in\ob (H_i)$, $i\ne j$, with $\Omega _{(^{i\!\!}A^j)}=\Omega _{A_j}$.
Since
\begin{equation*}
A_{x_1}^j=I_{H_1}\otimes\cdots\otimes I_{H_{j-1}}\otimes A_{jx_j}\otimes I_{H_{j+1}}\otimes\cdots\otimes I_{H_n}
\end{equation*}
we obtain
\begin{align*}
^{i\!\!}A_{x_j}^j&=\tfrac{1}{m_1\cdots m_{i-1}m_{i+1}\cdots m_n}\,\trace (I_{H_1})\cdots\trace (I_{H_{i-1}})\trace (I_{H_{i+1}})
   \cdots\trace (I_{H_n})\trace (A_{jx_j})I _{H_i}\\
   &=\tfrac{1}{m_j}\,\trace (A_{jx_j})I_{H_i}
\end{align*}
An observable of the form $B_{x_j}=\lambda _jI$ where $\lambda _j\in\sqbrac{0,1}$, $\sum\limits _j\lambda _j=1$ is called an \textit{identity observable} so
$^{i\!\!}A^j$ is an example of an identity observable. The $\rho$-distribution of $^{i\!\!}A^j$ is
\begin{equation*}
\Phi _{(\rho )}^{(^{i\!\!}A^j)}(x_j)=\tfrac{1}{m_j}\,\trace (A_{jx_j})\qquad\qed
\end{equation*}
\end{exam}

Generalizing what we did in Example~2, if $A\in\ob (H_1\otimes\cdots\otimes H_n)$ is an $n$-observable, we define the \textit{reduced} $n$-observable
$^{i\!\!}A\in\ob (H_i)$ with $\Omega _{(^{i\!\!}A)}=\Omega _A$ to be given by \eqref{eq22}.

If $H_,H_1$ are finite-dimensional Hilbert spaces, an \textit{operation} $\jscript$ from $H$ to $H_1$ is a trace non-increasing, completely positive, linear map
$\jscript\colon\lscript (H)\to\lscript (H_1)$ \cite{bgl95,dl70,hz12,nc00}. We denote the set of operations from $H$ to $H_1$ by $\oscript (H,H_1)$. If
$\jscript\in\oscript (H,H_1)$ preserves the trace we call $\jscript$ a channel \cite{hz12,hol82}. It can be shown that every $\jscript\in\oscript (H,H_1)$ has the form
$\jscript (B)=\sum\limits _{i=1}^nK_iBK_i^*$ where $K_i\colon H\to H_1$ are linear operators satisfying $\sum\limits _{i=1}^nK_i^*K_i\le I_H$
\cite{bgl95,hz12,kra83,nc00}. The operators $K_i$ are called \textit{Kraus operators} for $\jscript$. It is easy to show that $\jscript$ is a channel if and only if
$\sum\limits _{i=1}^nK_i^*K_i=I_H$. If $\jscript\in\oscript (H,H_1)$, its unique \textit{dual operation} $\jscript ^*\colon\lscript (H_1)\to\lscript (H)$ satisfies
$\trace\sqbrac{B\jscript ^*(C)}=\trace\sqbrac{C\jscript (B)}$ for all $B\in\lscript (H)$, $C\in\lscript (H_1)$ \cite{gud22}. It is easy to check that when
$\jscript (B)=\sum K_iBK_i^*$, then $\jscript ^*(C)=\sum K_i^*CK_i$ and $\jscript ^*\colon\escript (H_1)\to\escript (H)$. We say that $\jscript\in (H,H_1)$
\textit{measures} the effect $a\in\escript (H)$ if $\trace\sqbrac{\jscript (\rho )}=\trace (\rho a)$ for all $\rho\in\sscript (H)$. Since
\begin{equation*}
\trace\sqbrac{\rho\jscript ^*(I_{H_1})}=\trace\sqbrac{I_{H_1}\jscript (\rho )}=\trace\sqbrac{\jscript (\rho )}
\end{equation*}
for all $\rho\in\sscript (H)$, we see that $\jscript$ measures the unique effect $\jscript ^*(I_{H_1})$.

An \textit{instrument} $\iscript$ from $H$ to $H_1$ is a finite set of operations $\iscript =\brac{\iscript _x\colon x\in\Omega _\iscript}$ from $H$ to $H_1$ such that 
$\iscriptbar =\sum\limits _{x\in\Omega _\iscript}\iscript _x\in\oscript (H,H_1)$ is a channel. We call $\Omega _\iscript$ the \textit{outcome space} for $\iscript$ and denote the set of instruments from $H$ to $H_1$ by $\instr (H,H_1)$. If $H=H_1$, we write $\instr (H)$ for $\instr (H,H)$. For
$\Delta\subseteq\Omega _\iscript$ we write $\iscript (\Delta )=\sum\brac{\iscript _x\colon x\in\Delta}$ and call $\iscript$ an \textit{operation-valued measure} or a
\textit{positive operator-valued measure} (POVM) \cite{bgl95,dl70,hz12,nc00}. If $\iscript\in\instr (H,H_1)$, $\rho\in\sscript (H)$, the $\rho$-\textit{distribution} of
$\iscript$ is the probability measure on $\Omega _\iscript$ given by 
\begin{equation*}
\Phi _\rho ^\iscript (\Delta )=\trace\sqbrac{\iscript (\Delta )(\rho )}=\sum _{x\in\Delta}\trace\sqbrac{\iscript _x(\rho )}
\end{equation*}

An instrument $\iscript\in\instr (H,H_1)$ \textit{measures} a unique observable $\iscripthat\in\ob (H)$ given by $\Omega _{\iscripthat}=\Omega _\iscript$ where
$\trace (\rho\iscripthat _x)=\trace\sqbrac{\iscript _x(\rho )}$ for all $\rho\in\sscript (H)$. Since
$\trace\sqbrac{\iscript _x(\rho )}=\trace\sqbrac{\rho\iscript _x^*(I_{H_1})}$ for all $\rho\in\sscript (H)$, we see that $\iscripthat _x=\iscript _x^*(I_{H_1})$ for all
$x\in\Omega _{\iscripthat}$. The $\rho$-distribution of $\iscript$ becomes
\begin{equation*}
\Phi _\rho ^\iscript (\Delta )=\trace\sqbrac{\rho\iscript ^*(\Delta )(I_H)}=\trace\sqbrac{\rho\iscripthat\,(\Delta )}=\Phi _\rho ^{\iscripthat}(\Delta )
\end{equation*}
for all $\Delta\subseteq\Omega _\iscript$.

We now consider three examples of instruments. Let $\Omega _\kscript$ be a finite set and let $K_x\in\lscript (H,H_1)$, $x\in\Omega _\kscript$, with
$\sum\limits _{x\in\Omega _\kscript}K_x^*K_x=I$. The corresponding \textit{Krause instrument} $\kscript\in\instr (H,H_1)$ satisfies, $\kscript _x(B)=K_xBK_x^*$ and we call $K_x$ the \textit{Kraus operators} for $\kscript$ \cite{hz12,kra83}. Since $\kscript _x^*(C)=K_x^*CK_x$ for all $C\in\lscript (H_1)$ we conclude that the observable measured by $\kscript$ is
\begin{equation*}
\kscripthat _x=\kscript _x^*(I_{H_1})=K_x^*K_x
\end{equation*}
The $\rho$-distribution of $\kscript$ becomes
\begin{equation*}
\Phi _\rho ^\kscript (\Delta )=\trace\sqbrac{\rho\kscript (\Delta )}=\sum _{x\in\Delta}\trace\sqbrac{\kscript _x(\rho )}=\sum _{x\in\Delta}\trace (K_x\rho K_x^*)
   =\sum _{x\in\Delta}\trace (\rho K_x^*K_x)
\end{equation*}
As a special case of a Kraus instrument, let $A\in\ob (H)$ and define the instrument $\lscript ^{(A)}\in\instr (H)$ by $\lscript _x^{(A)}(B)=A_x^{1/2}BA_x^{1/2}$,
$x\in\Omega _A$. Then $\lscript ^{(A)}$ is called a L\"uders instrument \cite{gud120,lud51} and we see that $\lscript ^{(A)}$ is a Kraus instrument with Kraus operators $A _x^{1/2}$. The observable measured by $\lscript ^{(A)}$ is $A$. As a third example, let $A\in\ob (H)$ and let $\alpha _x\in\sscript (H_1)$,
$x\in\Omega _A$. We define the corresponding \textit{Holevo instrument} \cite{gud120,gud220,gud23,hol82}, $\hscript ^{(A,\alpha )}\in\instr (H,H_1)$ by
$\Omega _{\hscript ^{(A,\alpha )}}=\Omega _A$ and 
\begin{equation*}
\hscript _x^{(A,\alpha )}(B)=\trace (BA_x)\alpha _x
\end{equation*}
It follows that $\paren{\hscript _x^{(A,\alpha )^*}}(C)=\trace (C\alpha _x)A_x$. Hence
\begin{equation*}
\hscripthat _x^{(A,\alpha )}=\paren{\hscript _x^{(A,\alpha )}}^*(I_{H_1})=A_x
\end{equation*}
so $\hscript ^{(A,\alpha )}$ measures $A$.

As with observables, we define a \textit{multi-instrument} (or $n$-\textit{instrument}) $\iscript\in\instr (H,H_1)$ to have
$\Omega _\iscript =\Omega _1\times\cdots\times\Omega _n$ and we write the corresponding operations as $\iscript _{x_1\cdots x_n}$. If
$\iscript\in\instr (H,H_1)$ is an $n$-instrument, we define its $i$-\textit{marginal} to be the instrument $\iscript ^i\in\instr (H,H_1)$ given by
\begin{equation*}
\iscript _y^i(\rho )=\iscript\sqbrac{\Omega _\iscript ^{(x_i=y)}}(\rho )
\end{equation*}
for all $\rho\in\sscript (H)$, $i=1,2,\ldots ,n$. If $\iscript\in\instr (H,H_1\otimes\cdots\otimes H_n)$, we define the reduced instrument $^{i\!\!}\iscript\in\instr (H,H_1)$ with $\Omega _{^{(\i\!\!}A)}=\Omega _\iscript$ to be given by
\begin{equation*}
^{i\!\!}\iscript _x(\rho ) =\trace _{H_1}\cdots\trace _{H_{i-1}}\trace _{H_{i+1}}\cdots\trace _{H_n}\sqbrac{\iscript _x(\rho )}
\end{equation*}
Since
\begin{equation*}
\trace\sqbrac{\iscript _y^i(\rho )}=\trace\sqbrac{\iscript (\Omega _\iscript ^{(x_i=y)})(\rho )}=\trace\sqbrac{\rho\iscripthat (\Omega _\iscript ^{(x_i=y)})}
   =\trace (\rho\iscripthat _i^{\,i})
\end{equation*}
for all $\rho\in\sscript (H)$, we have that $(\iscript ^i)=\iscripthat^{\,i}$. Moreover, since $\trace\sqbrac{^{i\!\!}\iscript _x(\rho )}=\trace\sqbrac{\iscript _x(\rho )}$ we conclude that $(^{i\!\!}\iscript )^\wedge =\iscripthat$, $i=1,2,\ldots ,n$.

We say that a set of instruments $\iscript _i\in\instr (H,H_1)$, $i=1,2,\ldots ,n$
\textit{coexists} (is \textit{compatible}) if there exists a \textit{joint} $n$-instrument $\jscript\in\instr (H,H_1\otimes\cdots\otimes H_n)$ such that
$^{\i\!\!}\jscript ^i=\iscript _i$, $i=1,2,\ldots ,n$.

\begin{lem}    
\label{lem21}
If $\iscript _i\in\instr (H,H_i)$ coexist, then $\iscripthat _i\in\ob (H)$ coexist, $i=1,2,\ldots ,n$.
\end{lem}
\begin{proof}
To keep the notation simple, we assume that $n=2$ and the general result is similar. Assume that $\iscript _i\in\instr (H,H_i)$, $i=1,2$, coexist and
$\jscript\in\instr (H,H_1\otimes H_2)$ is a joint instrument. Then $^{1\!\!}\jscript _i^1=\iscript _{1x}$ and $^{2\!\!}\jscript _y^2=\iscript _{2y}$ for all
$x\in\Omega _{\iscript _1}$,
$y\in\Omega _{\iscript _2}$. Hence,
\begin{align*}
\iscript _{1x}(\rho )=\trace _{H_2}\sqbrac{\sum _y\jscript _{xy}(\rho )}\\
\iscript _{2y}(\rho )=\trace _{H_1}\sqbrac{\sum _x\jscript _{xy}(\rho )}
\end{align*}
for all $\rho\in\sscript (H)$. Now $\jscripthat\in\ob (H)$ is a bi-observable satisfying 
\begin{equation*}
    \trace (\rho\jscripthat _{xy})=\trace\sqbrac{\jscript _{xy}(\rho )}
\end{equation*}
for all $x\in\Omega _{\iscript _1}$, $y\in\Omega _{\iscript _2}$, $\rho\in\sscript (H)$. Since
\begin{equation*}
\trace\paren{\rho\sum _y\jscripthat _{xy}}=\trace\sqbrac{\sum _y\jscript _{xy}(\rho )}=\trace\sqbrac{\trace _{H_2}\sum _y\jscript _{xy}(\rho )}
   =\trace\sqbrac{\iscript _{1x}(\rho )}=\trace (\rho\iscripthat _{1x})
\end{equation*}
for all $\rho\in\sscript (H)$ we have that $(\jscripthat\,)_x^1=\iscripthat _{1x}$ and similarly $(\jscripthat\,)_y^2=\iscripthat _{2y}$. Hence,
$\iscripthat _1,\iscripthat _2$ coexist.
\end{proof}

It can be shown that the converse of Lemma~\ref{lem21} does not hold \cite{hz12,nc00}. An instrument $\jscript\in\instr (H,H_1)$ is \textit{part} of an instrument
$\iscript\in\instr (H,H_1)$ if there exists a surjection $f\colon\Omega _\iscript\to\Omega _\jscript$ such that \cite{gud220,hz12}
\begin{equation*}
\jscript _y=\iscript\sqbrac{f^{-1}(y)}=\sum\brac{\iscript _x\colon f(x)=y}
\end{equation*}
and we then write $\jscript =f(\iscript )$. For example, if $\hscript ^{(A,\alpha )}\in\instr (H,H_1)$ is Holevo with $\alpha =\brac{\alpha}$, then
\begin{align*}
f\paren{\hscript ^{(A,\alpha )}}_y(\rho )&=\sum\brac{\hscript _x^{(A,\alpha )}(\rho )\colon f(x)=y}=\sum\brac{\trace (\rho A_x)\alpha\colon f(x)=y}\\
   &=\trace\sqbrac{\rho\sum _x\brac{A_x\colon f(x)=y}}\alpha =\trace\sqbrac{\rho f(A)_y}\alpha =\hscript _y^{(f(A),\alpha )}(\rho )
\end{align*}
Hence, $f\paren{\hscript ^{(A,\alpha )}}=\hscript ^{(f(A),\alpha )}$. In general, we have
\begin{align*}
f(\iscript )_y^\wedge&=f(\iscript )_y^*(I_{H_1})=\sqbrac{\sum\brac{\iscript _x\colon f(x)=y}}^*(I_{H_1})=\sum\brac{\iscript _x^*\colon f(x)=y}(\iscript _{H_1})\\
   &=\sum\brac{\iscripthat _x\colon f(x)=y}=f(\,\iscripthat\,)_y
\end{align*}
so it follows that $f(\iscript )^\wedge =f(\,\iscripthat\,)$. It also follows that $f(\iscript )^*=f(\iscript ^*)$.

\section{Results}  
Our first result shows there us a close relationship between a nontrivial $n$-observable and its parts.

\begin{thm}    
\label{thm31}
$A$ is a nontrivial $n$-observable if and only if $A$ has $n$ parts $B_i=f_i(A)$ such that $\ab{f_i(\Omega _A)}\ge 2$ and
$\ab{\bigcap\limits _{i=1}^nf_i^{-1}(x_i)}=1$ for all $x_i\in\Omega _{B_i}$, $i=1,2,\ldots ,n$. In this case, $B_i$ is the $i$th marginal $B_i=A^i$, $i=1,2,\ldots ,n$.
\end{thm}
\begin{proof}
Suppose $A$ is a nontrivial $n$-observable in which case we can assume that $\Omega _A=\Omega _1\times\cdots\times\Omega _n$ where
$\ab{\Omega _i}\ge 2$, $i=1,2,\ldots ,n$. Let $B_i=A^i$, $i=1,2,\ldots ,n$, be the $i$th marginal. For $i=1,2,\ldots ,n$, define the surjective
$f_i\colon\Omega _A\to\Omega _i$ by $f_i(x_1,x_2,\ldots ,x_n)=x_i$. As in \eqref{eq21} we have
\begin{equation*}
f_i(A)(y_i)=A_{y_i}^i=B_{iy_i}
\end{equation*}
Hence, $B_i=A^i=f_i(A)$ is part of $A$, $i=1,2,\ldots ,n$. Also,
\begin{equation*}
\ab{\bigcap _{i=1}^nf_i^{-1}(\Omega _A)}=\ab{\brac{(x_1,\ldots ,x_n)}}=1
\end{equation*}
and $\ab{f_i(\Omega _A)}=\ab{\Omega _i}\ge 2$, $i=1,2,\ldots ,n$. Conversely, suppose $A$ has $n$ parts $B_i=f_i(A)$ such that $\ab{f_i(A)}\ge 2$ and
$\ab{\bigcap\limits _{i=1}^nf _i^{-1}(x_i)}=1$ for all $x_i\in\Omega _{B_i}$, $i=1,2,\ldots ,n$. Define
$h\colon\Omega _A\to\Omega _{B_1}\times\cdots\times\Omega _{B_n}$ by $h(x)=\paren{f_1(x),\ldots ,f_n(x)}$. If
$(x_1,\ldots ,x_n)\in\Omega _{B_1}\times\cdots\times\Omega _{B_n}$, since $\bigcap\limits _{i=1}^nf_i^{-1}(x_i)\ne\emptyset$ there exists an $x\in\Omega _A$ such that $f_i(x)=x_i$, $i=1,2,\ldots ,n$. Hence, $h(x)=(x_1,\ldots ,x_n)$ so $h$ is surjective. If $x,y\in\Omega _A$ with $h(x)=h(y)$, then $f_i(x)=f_i(y)$ for all $i=1,2,\ldots ,n$. Letting $x_i=f_i(x)=f_i(y)$ we obtain $x,y\in f_i^{-1}(x_i)$ for all $i=1,2,\ldots ,n$, so $x,y\in\bigcap\limits _{i=1}^nf_i^{-1}(x_i)$. Since
$\ab{\bigcap\limits _{i=1}^nf_i^{-1}(x_i)}=1$, we conclude that $x=y$. Hence, $h$ is injective so $h$ is bijective. Since $\ab{\Omega _{B_i}}=f_i(\Omega _A)\ge 2$, we have that $A$ is a nontrivial $n$-observable. Finally, we have the $i$th marginal
\begin{align*}
A_y^i&=\sum\brac{A_{x_1\cdots x_{i-1}yx_{i+1}\cdots x_n}\colon x_1,\ldots ,x_{i-1},x_{i+1},\ldots ,x_n}=\sum\brac{A_x\colon f_i(x)=y}\\
   &=f_i (A)_y=B_{iy}
\end{align*}
so $A_y^i=B_i$, $i=1,2,\ldots ,n$.
\end{proof}

A similar proof shows that Theorem~\ref{thm31} holds for $n$-instruments as well. If $\iscript _i\in\oscript (H_i,H'_i)$, $i=1,2,\ldots ,n$, we define the
\textit{tensor product} $\kscript =\iscript _1\otimes\cdots\otimes\iscript _n$ to be the operation
\begin{equation*}
\kscript\in\oscript (H_1\otimes\cdots\otimes H_n,H'_1\otimes\cdots\otimes H'_n)
\end{equation*}
that satisfies
\begin{equation*}
\kscript (B_1\otimes\cdots\otimes B_n)=\iscript _1(B_1)\otimes\cdots\otimes\iscript _n(B_n)
\end{equation*}
for all $B_i\in\lscript (H_i)$, $i=1,2,\ldots ,n$. It is not clear that the operation $\kscript$ exists. This is remedied by the next theorem.

\begin{thm}    
\label{thm32}
The operation $\kscript$ exists.
\end{thm}
\begin{proof}
Suppose that $\iscript _i$ has the Kraus decomposition $\iscript _i(B)=\sum\limits _{j_i}K_{j_i}^iBK_{j_i}^{i*}$ for all $B\in\lscript (H_i)$ and that
$\sum\limits _{j_i}K_{j_i}^{i*}K_{j_i}^i\le I_{H_i}$, $i=1,2,\ldots ,n$. Then for $C\in\lscript (H_1\otimes\cdots\otimes H_n)$ we define
\begin{equation*}
\kscript (C)=\sum _{j_1,\ldots ,j_n}K_{j_1}^1\otimes\cdots\otimes K_{j_n}^nCK_{j_1}^{1*}\otimes\cdots\otimes K_{j_n}^{n*}
\end{equation*}
Now $\kscript\in\oscript (H_1\otimes\cdots\otimes H_n,H'_1\otimes\cdots\otimes H'_n)$ because $\kscript$ has a Kraus decomposition with 
\begin{align*}
\sum _{j_1,\ldots ,j_n}(K_{j_1}^{1*}\otimes\cdots\otimes K_{j_n}^{n*})(K_{j_1}^1\otimes\cdots\otimes K_{j_n}^n)
   &=\sum _{j_1}K_{j_1}^{1*}K_{j_1}^1\otimes\cdots\otimes\sum _{j_n}K_{j_n}^{n*}K_{jn}^n\\
   &\le I_{H_1}\otimes\cdots\otimes I_{H_n}
\end{align*}
Finally, $\kscript$ satisfies
\begin{align*}
\kscript (B_1\otimes\cdots\otimes B_n)&=\sum _{j_1,\ldots ,j_n}K_{j_1}^1\otimes\cdots\otimes K_{j_n}^n(B_1\otimes\cdots\otimes B_n)
   K_{j_1}^{1*}\otimes\cdots\otimes K_{j_n}^{n*}\\
   &=\sum _{j_1,\ldots ,j_n}K_{j_1}^1B_1K_{j_1}^{1*}\otimes\cdots\otimes K_{j_n}^nB_nK_{j_n}^{n*}\\
   &=\sum _{j_1}K_{j_1}^1B_1K_{j_1}^{1*}\otimes\cdots\otimes\sum _{j_n}K_{j_n}^nB_nK_{j_n}^{n*}\\
   &=\iscript _1(B_1)\otimes\cdots\otimes\iscript _n(B_n)
\end{align*}
for all $B_i\in\lscript (H_i)$, $i=1,2,\ldots ,n$.
\end{proof}

If $\iscript _i\in\instr (H_i,H'_i)$, $i=1,2,\ldots ,n$, define the \textit{tensor product}
\begin{equation*}
\kscript\in\instr (H_1\otimes\cdots\otimes H_n,H'_1\otimes\cdots\otimes H'_n)
\end{equation*}
to be the $n$-instrument given by
\begin{equation*}
\kscript _{x_1\cdots x_n}(\rho )=\iscript _{x_1}\otimes\cdots\otimes\iscript _{x_n}(\rho )
\end{equation*}
for all $\rho\in\sscript (H_1\otimes\cdots\otimes H_n)$. We have seen that
\begin{equation*}
\kscript _{x_1\cdots x_n}\in\oscript (H_1\otimes\cdots\otimes H_n,H'_1\otimes\cdots\otimes H'_n)
\end{equation*}
and $\kscriptbar$ is a channel because $\kscriptbar =\iscriptbar _1\otimes\cdots\otimes\iscriptbar _n$ and $\iscriptbar _i$ are channels, $i=1,2,\ldots ,n$. The next result shows that $\iscript _i\otimes\cdots\otimes\iscript _n$ is a type of joint instrument for $\iscript _i$, $i=1,2,\ldots ,n$ even though it is not strong enough to provide coexistence.

\begin{thm}    
\label{thm33}
Let $\iscript _i\in\instr (H_i,H'_i)$, $i=1,2,\ldots ,n$, and let $\kscript =\iscript _1\otimes\cdots\otimes\iscript _n$.
{\rm{(i)}}\enspace $\kscripthat _{x_1\cdots x_n}=\iscripthat _{1x_1}\otimes\cdots\otimes\iscripthat _{nx_n}$.
{\rm{(ii)}}\enspace For all $\rho\in\sscript (H_1\otimes\cdots\otimes H_n)$ we have
\begin{align*}
^{1\!}\kscript _{x_1}^1(\rho )&=\iscript _{1x_1}\sqbrac{\trace _{H_2}\cdots\trace _{H_n}(\rho )}\\
   &\ \vdots\\
   ^{n\!}\kscript _{x_n}^n(\rho )&=\iscript _{nx_n}\sqbrac{\trace _{H_1}\cdots\trace _{H_{n-1}}(\rho )}
\end{align*}
{\rm{(iii)}}\enspace For all $\rho _i\in\sscript (H_i)$, letting $m_i=\dim H_i$ we have
\begin{align*}
\tfrac{1}{m_2\cdots m_n}^{1\!}\kscript _{x_1}^1(\rho _1\otimes I_{H_2}\otimes\cdots\otimes I_{H_n})&=\iscript _{1x_1}(\rho _1)\\
\vdots\hskip 10 pc&\ \vdots\\
\tfrac{1}{m_1m_2\cdots m_{n-1}}^{n\!}\kscript _{x_n}^n(I_{H_1}\otimes\cdots\otimes I_{H_{n-1}}\otimes\rho _n)&=\iscript _{nx_n}(\rho _n)
\end{align*}
\end{thm}
\begin{proof}
(i)\enspace For all $B=B_1\otimes\cdots\otimes B_n\in\lscript (H_1\otimes\cdots\otimes H_n)$ we obtain
\begin{align*}
\trace (B\kscripthat _{x_1\cdots x_n})&=\trace\sqbrac{\kscript _{x_1\cdots x_n}(B)}
   =\trace\sqbrac{\iscript _{1x_1}\otimes\cdots\otimes\iscript _{nx_n}(B_1\otimes\cdots\otimes B_n)}\\
   &=\trace\sqbrac{\iscript _{1x_1}(B_1)\otimes\cdots\otimes\iscript _{nx_n}(B_n)}=\trace\sqbrac{\iscript _{1x_1}(B_1)}\cdots\trace\sqbrac{\iscript _{nx_n}(B_n)}\\
   &=\trace (B_1\iscripthat _{1x_1})\cdots\trace (B_n\iscripthat _{nx_n})=\trace (B_1\iscripthat _{1x_1}\otimes\cdots\otimes B_n\iscripthat _{nx_n})\\
   &=\trace(B_1\otimes\cdots\otimes B_n\iscripthat _{1x_1}\otimes\cdots\otimes\iscripthat _{nx_n})
   =\trace (B\iscripthat _{1x_1}\otimes\cdots\otimes\iscripthat _{nx_n})
\end{align*}
Since any $A\in\lscript (H_1\otimes\cdots\otimes H_n)$ has the product form
\begin{equation*}
A=\sum _{i_1,\ldots ,i_n}(A_{i_1}\otimes\cdots\otimes A_{i_n})
\end{equation*}
with $A_{i_j}\in\lscript (H_j)$ the result holds for any $B\in\lscript (H_1\otimes\cdots\otimes H_n)$. Hence, (i) holds.\newline
(ii)\enspace As in (i), letting $B=B_1\otimes\cdots\otimes B_n\in\lscript (H_1\otimes\cdots\otimes H_n)$ we obtain
\begin{align*}
^{1\!}\kscript _{x_1}^1(B)&=\trace _{H'_2}\cdots\trace _{H'_n}\sqbrac{\sum _{x_2,\ldots ,x_n}\kscript _{x_1\cdots x_n}(B)}\\
   &=\trace _{H'_2}\cdots\trace _{H'_n}\sqbrac{\sum _{x_2,\ldots x_n}\iscript _{1x_1}\otimes\cdots\otimes\iscript _{nx_n}(B_1\otimes\cdots\otimes B_n)}\\
   &=\trace _{H'_2}\cdots\trace _{H'_n}\sqbrac{\sum _{x_2,\ldots ,x_n}\iscript _{1x_1}(B_1)\otimes\cdots\otimes\iscript _{nx_n}(B_n)}\\
   &=\trace _{H'_2}\cdots\trace _{H'_n}\sqbrac{\iscript _{1x_1}(B_1)\otimes\iscriptbar _2(B_2)\otimes\iscriptbar _n(B_n)}\\
   &=\iscript _{1x_1}(B_1)\trace _{H_2}(B_2)\cdots\trace _{H_n}(B_n)=\iscript _{1x_1}\sqbrac{\trace _{H_2}\cdots\trace _{H_n}(B_1\otimes\cdots\otimes B_n)}\\
   &=\iscript _{1x_1}\sqbrac{\trace _{H_2}\cdots\trace _{H_n}(B)}
\end{align*}
Since any $\rho\in\sscript (H_1\otimes\cdots\otimes H_2)$ has form as in (i) we conclude that (ii) holds.\newline 
(iii) Letting
\begin{equation*}
\rho =\rho _1\otimes\frac{I_{H_2}}{m_2}\otimes\cdots\otimes\frac{I_{H_n}}{m_n}
\end{equation*}
in (ii), we obtain
\begin{align*}
^{1\!}\kscript _{x_1}^1(\rho _1\otimes I_{H_2}\otimes\cdots\otimes I_{H_n})
   &=\iscript _{1x_1}\sqbrac{\trace _{H_2}\cdots\trace _{H_n}(\rho _1\otimes I_{H_2}\otimes\cdots\otimes I_{H_n})}\\
   &=\trace (I_{H_2})\cdots\trace (I_{H_n})\iscript _{1x_1}(\rho _1)\\
   &=m_2\cdots m_n\iscript _{1x_1}(\rho _1)
\end{align*}
Hence, the result holds for $\iscript _{1x_1}(\rho _1)$. The result holds for $\iscript _{2x_2}(\rho _2),\ldots ,\iscript _{nx_n}(\rho _n)$ in a similar way.
\end{proof}

If $\iscript _i\in\instr (H_i,H_{i+1})$, $i=1,2,\ldots ,n$, we define their \textit{sequential product} to be the $n$-instrument \cite{gn01,gud220,gud22}
\begin{equation*}
\iscript =\iscript _1\circ\cdots\circ\iscript _n\in\instr (H_1,H_{n+1})
\end{equation*}
given by $\Omega _\iscript =\Omega _{\jscript _1}\times\cdots\times\Omega _{\iscript _n}$ with
\begin{equation*}
\iscript _{x_1\cdots x_n}(\rho )=\iscript _{nx_n}\cdots\iscript _{1x_1}(\rho )
\end{equation*}
For example, if $\kscript _i\in\iscript (H_i,H_{i+n})$, $i=1,2,\ldots ,n$, are Kraus instruments with Kraus operators $K_{x_i}^i$, then 
\begin{equation*}
\kscript =\kscript _1\circ\ldots\circ\kscript _n\in\instr (H_1,H_{n+1})
\end{equation*}
is a Kraus $n$-instrument given by
\begin{equation*}
\kscript _{x_1\ldots x_n}(\rho )=K_{x_n}^n\ldots K_{x_1}^1\rho K_{x_1}^{1*}\ldots K_{x_n}^{n*}
\end{equation*}
so the Kraus operators for $\kscript$ are $K_{x_n}^n\ldots K_{x_1}^1$.

\begin{exam}{3}  
Let $\hscript ^{(A,\alpha )}\in\instr (H_1,H_2)$, $\hscript ^{(B,\beta )}\in\instr (H_2,H_3)$ be Holevo instruments and let
$\iscript =\hscript ^{(A,\alpha )}\circ\hscript ^{(B,\beta )}$ be their sequential product. Then for all $\rho\in\sscript (H)$ we have
\begin{align*}
\iscript _{x_1x_2}(\rho )&=\hscript _{x_2}^{(B,\beta )}\sqbrac{\hscript _{x_1}^{(A,\alpha )}(\rho )}
    =\hscript _{x_2}^{(B,\rho )}\sqbrac{\trace (\rho A_{x_1})\alpha _{x_1}}=\trace (\rho A_{x_1})\hscript _{x_2}^{(B,\rho )}(\alpha _{x_1})\\
    &=\trace (\rho A_1)\trace (\alpha _{x_1}B_{x_2})\beta _{x_2}=\trace\sqbrac{\rho \paren{\trace (\alpha _{x_1}B_{x_2})A_{x_1}}}\beta _{x_2}\\
    &=\hscript _{x_1x_2}^{(C,\beta )}(\rho )
\end{align*}
where $C_{x_1x_2}\in\ob (H_1)$ is the bi-observable given by $C_{x_1x_2}=\trace (\alpha _{x_1}B_{x_2})A_{x_1}$. We conclude that $\iscript$ is the Holevo instrument $\hscript ^{(C,\beta )}$ that measures the bi-observable $C$. It follows by induction that the sequential product of a finite number of Holevo instruments is Holevo.\quad\qedsymbol
\end{exam}

The sequential product $\iscript =\iscript _1\circ\cdots\circ\iscript _n$ measures the $n$-observable $\iscripthat\in\ob (H_1)$ given by
\begin{align*}
\iscripthat _{x_1\cdots x_n}&=\iscript _{x_1\cdots x_n}^*(I_{H_{n+1}})=\iscript _{1x_1}^*\paren{\iscript _{2x_2}^*\cdots\iscript _{nx_n}^*(I_{H_{n+1}})}\\
   &=\iscript _{1x_1}^*(\iscript _{2x_2}^*\cdots\iscript _{(n-1)x_{n-1}}(\,\iscripthat _n))
\end{align*}
For example, if $\iscript =\hscript ^{(A,\alpha )}\circ \hscript ^{(B,\beta )}$ as in Example~3, we obtain
\begin{align*}
\iscripthat _{x_1x_2}&=\iscript _{x_1x_2}^*(I_{H_3})=\hscript _{x_1}^{(A,\alpha )^*}\sqbrac{\hscript _{x_2}^{(B,\beta )^*}(I_{H_3})}
   =\hscript _1^{(A,\alpha )^*}(B_{x_2})\\
   &=\trace (\alpha _{x_1}B_2)A_{x_1}=C_{x_1}C_{x_2}
\end{align*}
as we know from Example~3. In the general case, the marginals for $\iscript$ become
\begin{align*}
\iscript _{x_1}^1&=\iscriptbar _n\paren{\,\iscriptbar _{n-1}\cdots\iscriptbar _2(\iscript _{1x_1}(\rho ))}\\
\iscript _{x_2}^2&=\iscriptbar _n\paren{\,\iscriptbar _{n-1}\cdots\iscriptbar _3\paren{\iscript _{2x_2}(\,\iscriptbar _1(\rho ))}}\\
   &\quad\vdots\\
   \iscript _{x_n}^n&=\iscript _{nx_n}\paren{\,\iscriptbar _{n-1}\cdots\iscriptbar _2\paren{\,\iscriptbar _1(\rho )}}
\end{align*}

We can also define sequential products of observables $A_i\in\ob (H)$, $i=1,2,\ldots n$. For example, we define the \textit{L\"uders sequential product} to be the $n$-observable \cite{gn01,gud220,gud22}
\begin{equation*}
(A_1\circ\cdots\circ A_n)_{x_1\cdots x_n}=(A_{1x_1})^{1/2}\cdots (A_{(n-1)x_{n-1}})^{1/2}A_{nx_n}(A_{(n-1)x_{n-1}})^{1/2}\cdots (A_{1x_1})^{1/2}
\end{equation*}
In the case when $n=2$, we have $(A_1\circ A_2)_{x_1x_2}=A_{1x_1}^{1/2}A_{2x_2}A_{1x_1}^{1/2}$. The marginals become $(A_1\circ A_2)_{x_1}^1=A_{x_1}$ and
\begin{equation*}
(A_1\circ A_2)_{x_2}^2=\sum _{x_1}A_{1x_1}^{1/2}A_{2x_x}A_{1x_1}^{1/2}
\end{equation*}
For simplicity, in the remainder of this article, we assume that $n=2$ and the generalization to arbitrary $n$ will be clear.

If $A,B\in\ob (H)$ and $\iscript\in\instr (H,H_1)$ measures $A$, we define the $\iscript$-\textit{sequential product of} $A$ \textit{then} $B$ to be the bi-observable $\paren{A\sqbrac{\iscript}B}_{xy}=\iscript _x^*(B_y)$ \cite{gud120,gud220,gud22}. Moreover, we define $B$ \textit{conditioned by} $A$ \textit{related to}
$\iscript$ \cite{gud120,gud220} by
\begin{equation*}
(B\mid\iscript\mid A)_y=\paren{A\sqbrac{\iscript}B}_y^2=\sum _x\iscript _x^*(B_y)
\end{equation*}
Since the 1-marginal is $\paren{A\sqbrac{\iscript}B}_x^*=A_x$ we see that $(B\mid\iscript\mid A)$ and $A$ coexist with joint bi-observable $A\sqbrac{\iscript}B$. For the L\"uders instrument $\lscript ^{(A)}$ we have
\begin{equation*}
\paren{A\sqbrac{\lscript ^{(A)}}B}_{xy}=\lscript _x^{(A)^*}(B_y)=A_x^{1/2}B_yA_x^{1/2}=(A\circ B)_{xy}
\end{equation*}
For the Kraus instrument $\kscript _x(\rho )=K_x\rho K_x^*$ we obtain
\begin{equation*}
\paren{A\sqbrac{\kscript}B}_{xy}=K_x^*(B_y)=K_x^*B_yK_x
\end{equation*}
In the case of Holevo instrument $\hscript ^{(A,\alpha )}$ we have
\begin{equation*}
\paren{A\sqbrac{\hscript ^{(A,\alpha )}}B}_{xy}=\hscript _x^{(A,\alpha )^*}(B_y)=\trace (\alpha _xB_y)A_x
\end{equation*}
For a futher discussion on this, we refer the reader to \cite{gud120,gud220,gud22,gud23}.

\end{document}